\documentclass{article}

\usepackage{kr}

\usepackage{times}
\usepackage{soul}
\usepackage{url}
\usepackage[hidelinks]{hyperref}
\usepackage[utf8]{inputenc}
\usepackage[small]{caption}
\usepackage{graphicx}
\usepackage{amsmath}
\usepackage{amsthm}
\usepackage{booktabs}
\usepackage{algorithm}
\usepackage{framed}
\usepackage{wrapfig}

\usepackage{amssymb}
\usepackage{framed}
\usepackage{DotArrow}
\usepackage{underscore}
\usepackage[shortlabels]{enumitem}
\usepackage{hyperref}
\usepackage{algpseudocodex}
\usepackage{tikz}
\usetikzlibrary{matrix}
\usetikzlibrary{decorations.pathreplacing,positioning,shadows}

\usepackage{todonotes}

\usepackage{xspace}
\newcommand{\alc}{\ensuremath{\mathcal{ALC}}\xspace}
\newcommand{\pure}{pure \flbot}
\newcommand{\el}{\ensuremath{\mathcal{EL}}\xspace}
\newcommand{\flo}{\ensuremath{\mathcal{FL}_0}\xspace}
\newcommand{\flbot}{\ensuremath{\mathcal{FL}_\bot}\xspace}
\newcommand{\eg}{e.g.\ }
\newcommand{\wrt}{w.r.t.\ }
\newcommand{\ie}{i.e.\ }
\newcommand{\Start}{\ensuremath{\Gamma_{start}}\xspace}
\newcommand{\Flat}{\ensuremath{\Gamma_{flat}}\xspace}

\newcommand{\pred}{\ensuremath{\mathbf{Var}}\xspace}
\newcommand{\fun}{\ensuremath{\mathbf{R}}\xspace}
\newcommand{\names}{\ensuremath{\mathbf{N}}\xspace}
\newcommand{\const}{\ensuremath{\mathbf{C}}\xspace}

\newcommand{\s}{\ensuremath{\mathfrak{S}}\xspace}
\newcommand{\Sh}{\ensuremath{\mathcal{S}}\xspace}
\newcommand{\Smaller}{\ensuremath{\mathcal{P}}\xspace}

\newcommand{\inibot}{\ensuremath{s^\bot_{ini}}\xspace}
\newcommand{\inicons}{\ensuremath{s^A_{ini}}\xspace}

\newcommand{\flreg}{
	\ensuremath{\mathcal{FL}_{reg}}\xspace}
\newcommand{\flregbot}{
	\ensuremath{\mathcal{FL_\bot}_{reg}}\xspace}

\newtheorem{definition}{Definition}
\newtheorem{lemma}{Lemma}
\newtheorem{example}{Example}

\newtheorem{theorem}{Theorem}
\newtheorem{claim}{Claim}

\newtheorem*{Rule}{Decreasing rule}

\author{%
		Barbara Morawska$^1$\and
		Dariusz Marzec$^{1}$ \\
		\affiliations
		$^1$Institute of Computer Science, University of Opole, Poland
		\emails
		\{barbara.morawska,  dariusz.marzec\}@uni.opole.pl
		\thanks{\small	
						\begin{wrapfigure}[4]{r}{0.10\textwidth}
								\vspace{-5mm}
								\hspace{-6mm}
								\includegraphics[width=2cm]{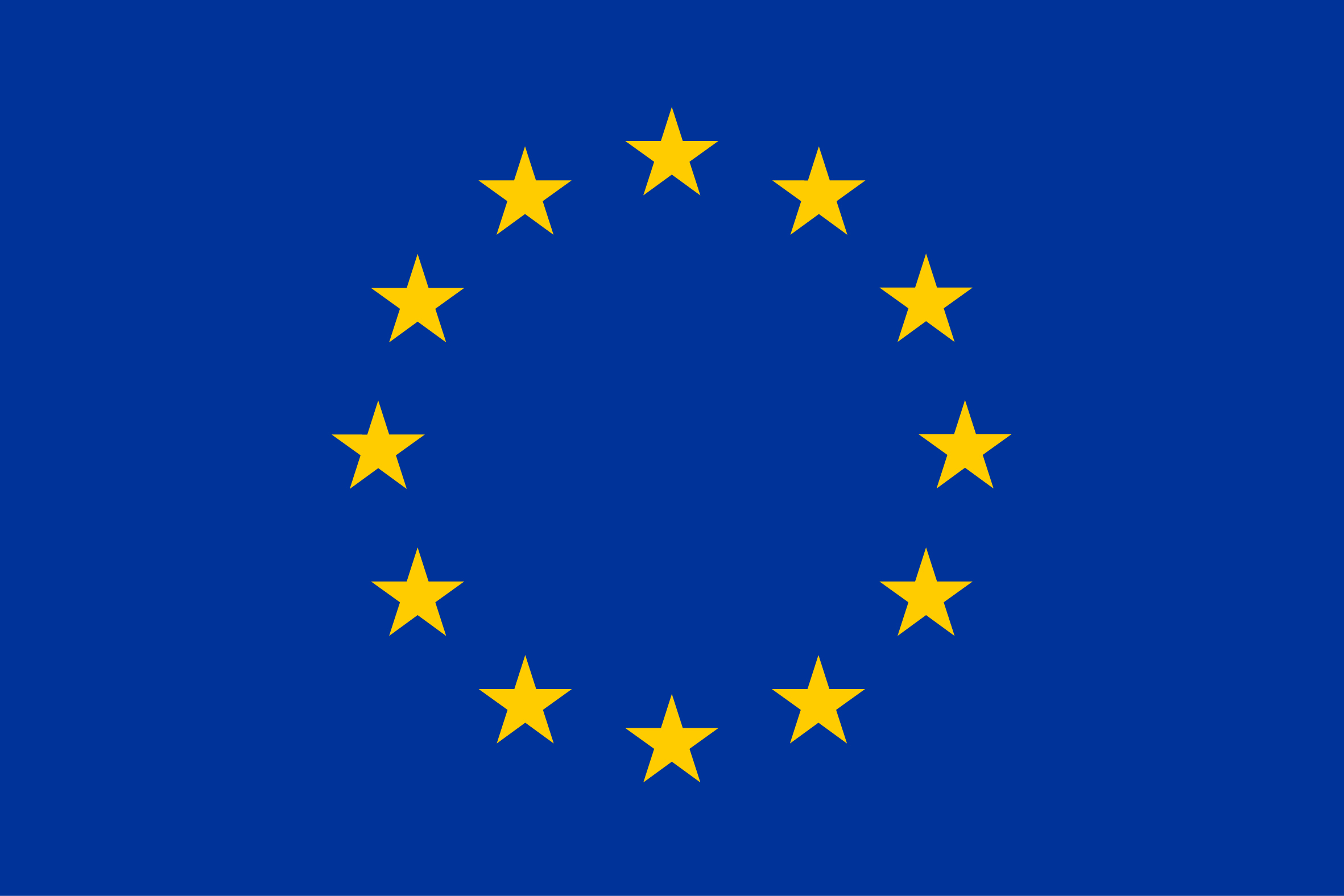}
							\end{wrapfigure}
						This research is part of the project No 2022/47/P/ST6/03196 within the POLONEZ BIS
						programme co-funded by the National Science Centre and the European Union’s Horizon 2020
						research and innovation programme under the Marie Skłodowska-Curie grant agreement
						No. 945339. For the purpose of Open Access, the author has applied a CC-BY public copyright
						licence to any Author Accepted Manuscript (AAM) version arising from this submission.		
				}
			}

\title{Solving unification in the description logic \flbot}

\begin{document}
	\maketitle
	\begin{abstract}
		We present an algorithm for solving the unification problem in the description logic \flbot.
		This logic extends \flo with the bottom constructor, and thus supports
		conjunction, value restrictions, top and bottom constructors.  
		Unification of concepts can be a useful tool for  ontology maintenance; however, little is  known
		about unification even in small, restricted description logics.
		The unification problem has been solved only for \flo and \el.
		This paper contributes to the ongoing effort to extend these results to richer logics.
		Our algorithm runs in exponential time with respect to the size of the problem.
	\end{abstract}
	
\section{Introduction}
Description Logics (DLs) are designed as a formal framework to represent and manipulate information stored in the form of an ontology -- a set of definitions of  concepts using simple concepts called \emph{names} and complex concepts.
Complex concepts are constructed from
primitive names using relations and constructors. There are many DLs that differ in the constructors they provide for building complex concepts out of simple ones.

Unification of concepts in DLs was proposed in \cite{BaNa-JSC01} as a non-standard reasoning service that can assist in maintaining large ontologies by identifying and eliminating redundancies.
%

%

Nevertheless, little is known about unification in DLs. The problem has been solved and is well understood only for sub-Boolean DLs such as \el and \flo. In this work, we focus on  the unification of concepts in an extension of \flo, namely \flbot.

The description logic \flbot allows the construction of complex concepts from a countable set of concept names and role names using conjunction, the top ($\top$) and bottom ($\bot$) constructors, and a value restriction of the form $\forall r.C$ where $r$ is a relation with the value of the second argument restricted to a concept $C$.
The description logic \flo provides the same constructors except for $\bot$.
Both of these lightweight DLs  have polynomial-time subsumption algorithms. 

 In order to define the unification problem in a restricted description logic, we first designate some concept names as variables. Then, given a pair of concepts $C$ and $D$ that may contain variables,
we seek a substitution that makes $C$ and $D$ equivalent (or makes one subsumed by the other).
For example, consider the pair of concepts:\\
\texttt{Viral_disease}$\sqcap\,\forall$\texttt{has_attribute.Infectious},\\
\texttt{Infectious\_disease}$\, \sqcap\, \forall$\texttt{has_cause.Virus}.
These concepts can be made equal by treating \texttt{Infectious\_disease} and \texttt{Viral_disease} as  variables which can be substituted by other concepts or in other words, further specified.
Then a unifier will be a mapping, or a set of definitions:
\begin{itemize}
	\item \texttt{Infectious\_disease}$:=$\texttt{Disease}$\sqcap$ $\forall$\texttt{\mbox{has\_attribute}.Infectious} 
	\item \texttt{Viral_disease}$:=$\texttt{Disease}$\sqcap\forall$\texttt{\mbox{has\_cause}.Virus}.
\end{itemize}

As illustrated in the above example, the substitution for variables that we aim to obtain from a unification algorithm consists of a set of definitions for previously undefined concept names.
Nevertheless, in what follows, we will treat the unification problem primarily as a decision problem: “Does such a substitution exist?” If it does, an example solution can then be recovered.

Unification of concepts was first studied and solved for the description logic \flo, \cite{BaNa-JSC01}, where the problem was shown to be ExpTime-complete.The unification algorithm presented in that work proceeds by first reducing the unification problem to the problem of solving formal language equations, and then further reducing it to an emptiness test for a tree automaton.

Using similar methods, the unification problem was later solved for the description logic \flreg  in \cite{Baader2001} and for \flregbot in \cite{Baader2002}. The description logic \flreg extends \flo by allowing regular expressions constructed over the set of role names, in place of simple role names, for defining value restrictions. The description logic \flregbot further extends \flreg with the $\bot$ constructor. In the same paper, the authors admitted that the unification in \flbot does not yield to their method. In the latter work, however, the authors noted that unification in \flbot does not yield to their method.

Since then, there have been many--mostly unpublished--attempts to solve the unification problem for \flbot. In \cite{Borgwardt}, a new approach to unification in \flo\ was presented. In \cite{Morawska2021}, an attempt was made to extend this approach to unification in \flbot. However, this attempt was flawed: it appears to be impossible to reduce unification in \flbot\ to that in \flo, contrary to the conjecture made in that work.
Here we present the algorithm which extends the unification in \flo to that in \flbot, but not by reduction. Instead, we provide a generalization of the unification approach for \flo from \cite{Borgwardt}. Unification in \flo is thus a special case of the proposed algorithm.
The paper is organized as follows. In Section~\ref{section:flbot}, we formally present the logic \flbot. The unification problem is defined in Section~\ref{section:unification-problem}. Then, in Sections~\ref{section:overview} (overview), \ref{section:normalization} (preprocessing), \ref{section:shortcuts}, and \ref{section:main} (main computation), we present the algorithm.
In the next three sections (\ref{section:termination}, \ref{section:soundness}, and \ref{section:completeness}), we prove the correctness of the algorithm. The paper concludes with final remarks in
Section~\ref{section:conclusions}.

\section{The description logic \flbot}\label{section:flbot}
The concepts in \flbot are generated by the following grammar from a finite set of concept names \names and a finite set of role names \fun:
\[
C ::= A \mid C \sqcap C \mid \forall r.C \mid \top \mid \bot,
\]
where $A \in \names$ and $r \in \fun$. 

The concepts are interpreted as subsets of a non-empty domain. An interpretation is a pair consisting of a domain and an interpreting function, $(\Delta^I, \cdot^I)$. A concept name $A$ is interpreted by a subset of the domain, $A^I \subseteq \Delta^I$, a role name $r$ is interpreted as a binary relation, $r^I \subseteq \Delta^I \times \Delta^I$, and a concept conjunction is interpreted by intersection: $(C_1 \sqcap C_2)^I = C_1^I \cap C_2^I$, value restriction $\forall r.C$ is interpreted as the set:
$(\forall r.C)^I = \{e \in \Delta^I\mid \forall d \in \Delta^I((e,d) \in r^I \implies d \in C^I)\}$; $\top^I = \Delta^I$ and $\bot^I = \emptyset$.

Based on this semantics, we define the equivalence and subsumption relations between concepts as follows: $C \equiv D$ iff $C^I = D^I$ and
$C \sqsubseteq D$ iff $C^I \subseteq D^I$, in every interpretation $I$.

The \textbf{subsumption (equivalence) problem} is then defined as follows.\\
\textbf{Input:} a pair of concepts $C, D$.\\
\textbf{Output:}\emph{``Yes"}, if  $C \sqsubseteq D$ ( $C \equiv D$), and \emph{``No"} otherwise.\\
In deciding the subsumption or equivalence problem between concepts in \flbot, it is convenient to bring them first into a normal form. 

A concept is in \textbf{normal form} if it is a conjunction of concepts of the form
$\forall r_1.(\forall r_2.( \dots \forall r_n. A))$, where $A$ is  a concept name or $\top$ or $\bot$. The
empty conjunction is equivalent to $\top$. 

Since for all concepts $C_1$ and $C_2$, the equivalence  $\forall r.(C_1 \sqcap C_2) \equiv \forall r.C_1 \sqcap \forall r.C_2$ holds in \flbot, each \flbot-concept is equivalent to a concept in normal form. 	

Concepts of the form $\forall r_1.(\forall r_2.(\dots \forall r_n. A)\dots)$ are called  \emph{particles}.
We identify a conjunction of particles with the set of its conjuncts.
In other words, every concept in normal form is a set of particles, and the empty set of particles corresponds to $\top$.

For brevity, a particle of the form $\forall r_1.(\forall r_2.(\dots \forall r_n. A)\dots)$ will be written as $\forall v.A$, where $v = r_1r_2\dots r_n$, and
$A$ may only be a concept name, $\top$ or $\bot$.
The symbol  `$v$' is a word over $\fun$, and we call it a \emph{role string}. We say that the size of a particle is the length of its role string, $|v|$. Relation $\le$ on role strings denotes the prefix relation, and $<$ denotes the proper prefix relation.

\begin{definition}(prefix)\label{definition:prefix}
We say that a particle $P$ is a \emph{prefix} of $P'$ in the following sense.
\begin{itemize}
	\item $P = \forall v. \bot$, $P' = \forall v'.\bot$ and $v < v'$ or
	\item  $P = \forall v. \bot$, $P' = \forall v'.A$ and $v \le v'$, where $A$ is a concept name.
	In particular, $\bot$ is a prefix of $A$.
\end{itemize}
\end{definition}

We call a particle of the form $\forall v.\top$, a $\top$-particle (top-particle), a particle of the form $\forall v.\bot$, a $\bot$-particle (bottom-particle), and a particle of the form
$\forall v.A$, an $A$-particle, for a concept name $A$. Notice that only a $\bot$-particle may be a prefix of another particle.

Since $\forall v.\top\equiv \top$, any $\top$-particle occurring in a concept $C$ may be simply deleted.
Since $\forall v .\bot \sqsubseteq \forall v'.A$, where $v \le v'$ and $A$ is a concept name or $\bot$, if a concept $C$ contains a particle 
$\forall v.\bot$, then for every $v'$ such that $v \le v'$, any particle $\forall v'.A$ may likewise be deleted from $C$.
Since $\bot \sqsubseteq D$ for any concept $D$, if a concept $C$ contains $\bot$ as a particle ($\bot \in C$), all other particles in $C$ may be deleted.

If there are no deletions applicable to a concept $C$, we say that $C$ is \textbf{reduced}.
Below we consider only reduced concepts in normal form.

For concepts in normal form we have a simple polynomial-time way of deciding the subsumption problem.
This follows directly from the following lemma.
\begin{lemma}(Characterization of subsumption in \flbot)\label{lemma:characterization}
Let $C , D$ be an instance of a subsumption problem, where $C, D$ are concepts in normal form:
$C=\{P_1, P_2, \dots, P_m\}$ and
$D=\{P'_1, P'_2, \dots, P'_n\}$.  Then $C \sqsubseteq D$ if and only if
for every $P'_i \in D$, one of the two conditions holds:
\begin{enumerate}
	\item\label{subsumption:identity} $P'_i \in C$ 
	\item\label{subsumption:prefix} $P'_i = \forall v.A$, where $A$ is a constant or $\bot$, and there exists a bottom particle $\forall v'.\bot$  (where $v'$ may be empty) in $C$ such that
	$v'\le v$. 
\end{enumerate}
\end{lemma}

\begin{proof}The lemma is justified by the properties of conjunction with respect to concepts in normal form.
	The first statement follows from the fact that each concept is a set of particles, and subsumption is a partial order on such sets.
	The second statement follows from the fact that $\bot$ is always interpreted as the empty set; hence, $\bot \sqsubseteq C$ for every concept $C$.
	This property is preserved for value restrictions constructed by prepending $\bot$ and $C$ with a role string $v$: $\forall v.\bot \sqsubseteq \forall v.C$.
\end{proof}
	When deciding a possible subsumption $\{P_1, \dots, P_n\} \sqsubseteq \forall v.A$, where $A$ is a concept name,
	we can ignore any $P_i$ involving a concept name other than $A$. The reason is that any particle with a concept name different
	than $A$ cannot affect passing any of the tests in Lemma~\ref{lemma:characterization}.

	Similarly, when deciding a possible subsumption $\{P_1, \dots, P_n\} \sqsubseteq \forall v.\bot$, we can ignore any $P_i$ not containing $\bot$,
	since according to the subsumption tests in Lemma~\ref{lemma:characterization} we check only for $\forall v.\bot \in \{P_1, \dots, P_n\} $ or for $ \forall v'.\bot \in \{P_1, \dots, P_n\} $, where $v' \le v$.

	\section{Unification problem}\label{section:unification-problem}

	In order to define the unification problem, we partition the set of concept names \names into two  disjoint sets $\const$ and $\pred$,  called
	constants and  variables, respectively. Variables are concept names that can be 
	substituted by other concepts.

	A substitution is a mapping from variables to \flbot-concepts constructed over a finite set of variables and constants and it is extended to all concepts in a usual way. 
	If  a variable $X$ is assigned $\top$ by a substitution $\gamma$ and  $\gamma$ is clear from the context, we call $X$ a $\top$-variable, and if it is assigned $\bot$, we call it a $\bot$-variable.
	
	The input for unification problem is defined as a set of pairs of \flbot-concepts, $\{(E_1, F_1), \dots, (E_n, F_n)\}$. We call such pairs \emph{possible subsumptions} or \emph{goal subsumptions} for intuitive reasons. The formulation below uses the symbol $\sqsubseteq^?$ instead of a comma, in order to indicate what we want to achieve in a solution.\\	
	
	\noindent
	\textbf{The unification problem}\\
	\textbf{Input:} $\Gamma = \{E_1 \sqsubseteq^? F_1, \dots, E_n\sqsubseteq^? F_n\}$,
	where $E_1, \dots, E_n, F_1, \dots F_n$ are \flbot concepts in normal form, and $F_1, \dots, F_n$ are particles. The concepts may contain variables.
	
	\noindent
	\textbf{Output:} "Yes" if there exists a substitution $\gamma$ (called a solution or a unifier of $\Gamma$) such that
	$\gamma(E_1) \sqsubseteq \gamma(F_1), 
	\dots, \gamma(E_n) \sqsubseteq \gamma(F_n)$; otherwise, "No".\\

	%
	%
	
	In deciding unification, it is sufficient to consider ground unifiers.
	A substitution $\gamma$ is ground if it assigns to variables concepts that contain no variables.
Based on our earlier observations about the subsumption problem, we may further assume that the unification problem involves at most one constant.
If multiple constants occur, we can decompose the problem into several subproblems, one for each constant  $A_1, \dots, A_n$ in  $\Gamma$.
For each $A_i$, let $\Gamma_{A_i}$ be the problem obtained from $\Gamma$ by replacing all constants other than $A_i$ with $\top$.
To solve $\Gamma$, we solve $\Gamma_{A_1}, \dots, \Gamma_{A_n}$ separately.
If $\gamma_1, \dots, \gamma_n$ are solutions to these subproblems, then $\gamma = \gamma_1 \cup \cdots \cup \gamma_n$ is clearly a solution to $\Gamma$.
%
%


%

%
	\begin{example}\label{example:unification}
		Let $\Gamma = \{\forall r.B \sqcap X \sqsubseteq^? \forall r.X \sqcap A\}$
		
		The problem is divided into two sub-problems.
		$\Gamma_A = \{X \sqsubseteq^? \forall r.X, X \sqsubseteq^? A\}$,
		$\Gamma_B = \{\forall r.B \sqcap X  \sqsubseteq^? \forall r.X\}$.
		
	Let $\gamma_A = [X \mapsto \bot]$ and $\gamma_B = [X \mapsto B]$.
Then $\gamma = \gamma_A \cup \gamma_B = [X \mapsto (\bot \sqcap B) \equiv \bot]$ is a solution to $\Gamma$.
	\end{example}
	
%
%
%
%

	
	\section{Overview of the algorithm}\label{section:overview}
	
Our unification algorithm consists of two stages: normalization and solving the normalized problem.

	\begin{enumerate}
		\item Normalization is a non-deterministic polynomial-time procedure.
		Some unification problems can already be determined to be unifiable at this stage; however, certain non-deterministic choices may lead to failure for the current branch.
		If all possible choices lead to failure, the algorithm terminates with failure as the final answer -- the problem is not unifiable.
		
		\item If the problem passes the first stage, we obtain a partial substitution (for variables set to $\bot$ or $\top$) and a unification problem in a special form, consisting of subsumptions of the type shown in Figure \ref{figure:subsumptions}.
		If there are no start subsumptions, the problem has a solution assigning $\top$ to all variables.
		If there are no flat subsumptions, the problem has a solution. The construction is explained in Figure~\ref{figure:construction}.
		
		If the set of flat unsolved subsumptions is non-empty, the algorithm enters the next stage, namely \emph{computing with shortcuts} (Section~\ref{section:shortcuts}).
		At this stage, our algorithm can fail for the current choices in the normalization stage, or 
		it can return success. If it returns success, this means that the problem is unifiable. We can recover an example of a unifier from the computations.
	\end{enumerate}
	
	\section{Normalization}\label{section:normalization}
	
	The input for the algorithm is $\Gamma = \{E_1 \sqsubseteq^? P_1, \dots, E_n \sqsubseteq^? P_n\}$,
	where $E_1, ..., E_n$ are sets of particles and $P_1, \dots, P_n$ are particles that can contain variables or are $\bot$-particles or A-particles where $A$ is a constant. We assume that $\Gamma$ contains only one constant $A$. Since all concepts are reduced, $\top$ does not occur in $\Gamma$.

	Initially all goal subsumptions are marked as \emph{unsolved}.
	
	
	\begin{enumerate}
		\item Initial steps. For each variable $X$  we have to guess if in a solution:
		 $X$ should be $\top$, or $X$ should be $\bot$, or if $X$ should be neither $\top$ nor $\bot$, but it should contain a constant $A$.
		 The last guess is expressed by maintaining a Boolean variable $A_X$, which should be true if $X$ contains $A$ under a solution or 
		 false in the opposite case. If $X$ is guessed to be $\bot$ or $\top$, then it is replaced by 
		$\bot$ or $\top$, respectively, everywhere in the goal subsumptions.
%
%
		\item Implicit Solver. We exhaustively apply the rules from Figure~\ref{figure:implicit}--in the given order--to the goal subsumptions. 
		Implicit Solver may detect unifiability of a given problem at this  stage or  detect failure for the choices made for each variable in the first step.

		\begin{figure}[h]
			\begin{framed}
				\textbf{Implicit Solver rules:}

				\begin{enumerate}[ref=\theenumii]
					\item\label{rule-1} If $\bot$ is found in a subsumption $s$ on its left side at top level, then we label $s$ as solved.
					\item\label{rule-2} If $\bot$ is found on the right side of a subsumption $s$, and not on the left side, then return \textbf{failure}. 
					\item\label{rule-3} If a $\top$-particle is found in a subsumption $s$ on its right side, then we label $s$ as solved.	
					\item\label{new3a} If $\top \sqsubseteq^? X$ is an unsolved subsumption and $X$ is not chosen to be $\top$, then return \textbf{failure}. 			
					\item If a $\top$-particle is found in a subsumption $s$ on its left side, we delete
					this particle from $s$.
					\item If a particle $P$ occurs on the right side of $s$ and also on the left side (at the top level), then label $s$ as solved.
					\item\label{new1} If the constant $A$ occurs on the right side of $s$ and there is a variable $X$ on the left side, such that $A_X$ is true, then label $s$ as solved.
					\item If the constant $A$ occurs on the right side of $s$, but on the left side of $s$ there is neither  $A$ nor a variable $X$ with $A_X$ true, then return \textbf{failure}.
					\item If $A$ occurs on the left side of a subsumption $s$ and $X$ is on its right side, where $A_X$ is false,
					then delete this occurrence of $A$. 
					\item\label{new2} If  $A_X$ is true and $X$ occurs on the right side of $s$, but neither $A$ nor a variable $Y$ with $A_Y$ true occurs on the left hand side of $s$, then return \textbf{failure}.
					\item If there are no unsolved subsumptions left, return \textbf{success}.
				\end{enumerate}
			\end{framed}
			\caption{Implicit Solver rules}
			\label{figure:implicit}
		\end{figure}
		Note that after applying rule~(\ref{rule-1}) and (\ref{rule-2}) of Figure~\ref{figure:implicit}, all variables that are guessed to be $\bot$ are  eliminated from the goal. Hence in particular there is no $\bot$ on the left side of any goal subsumption after applying these rules.
		
		\item Normalization loop. At this point, the algorithm enters a loop in which it alternates between applying the rules from Figure~\ref{figure:flattening} and Figure~\ref{figure:implicit}.
		
		\begin{enumerate}
			
			\item Flattening (Applying rules from Figure~\ref{figure:flattening}). We choose ('don't care' nondeterministically) an unsolved non-flat subsumption from the unification problem.

				An unsolved goal subsumption $C_1 \sqcap \cdots \sqcap C_n \sqsubseteq^? P$ is \textbf{non-flat} if 
				\begin{itemize}
					\item $P = \forall r.P'$, where $P'$ is a particle  or 
					\item there is $i$, $1 \le i \le n$ such that $C_i = \forall r.C_i'$, where $C_i'$ is a  particle.
					\item there is  $i$, $1 \le i \le n$ such that $C_i = A$.
				\end{itemize}
				
			We flatten the subsumption by applying a rule from Figure~\ref{figure:flattening}.
			
					\begin{figure}[h]
					\begin{framed}
						Consider a \textbf{non-flat, unsolved} goal subsumption,\\ $s=C_1 \sqcap \cdots \sqcap C_n \sqsubseteq^? P$.
						\begin{enumerate}[ref=\theenumii]
							\item If $P$ is of the form $\forall r.P'$,
							replace $s$ with $s^{-r}$.	
							\item\label{new} If $P$ is a variable $X$,
							replace $s$ by the following set of goal subsumptions:				
								$\{s^{-r} \mid r \in \fun \}$,
								and if $A_X$ true, add
								$C_1^{A} \sqcap \cdots \sqcap C_n^{A} \sqsubseteq^? A$ to the goal.
							\end{enumerate}		
						\end{framed}
						\caption{Flattening rules}	
						\label{figure:flattening}
					\end{figure}

				In Figure~\ref{figure:flattening} we use the following notation. 
				If $P$ is a particle and $r$ a role name ($r \in \fun$), we define an operation on particles $ P^{-r}$ in the following way:\\
				$
				P^{-r} = \begin{cases}
					P^{r} & \text{ if } P \text{ is a  variable and }\\
					&  P^r \text{ its \emph{decomposition variable}, }\\
					P' & \text{ if } P = \forall r.P', \\
					\top & \text{ in all other cases. } 
				\end{cases}
				$
				
				If $s$ is a  subsumption of the form $C_1 \sqcap \cdots \sqcap C_n \sqsubseteq^? D$, where
				$C_1, \dots, C_n, D$ are particles, we define
				$s^{-r} = C_1^{-r} \sqcap \cdots \sqcap C_n^{-r} \sqsubseteq^? D^{-r}$.
				
				
				Let $P$ is a particle, then we define $P^{A}$, for a constant $A$ in the following way:\\
				$
				P^{A} = \begin{cases}
					P & \text{ if } P \text{ is $A$  or a variable,}\\
					\top & \text{ in all other cases.} 
				\end{cases}
				$
				
				\textbf{Decomposition variables} The flattening rules applied to unsolved and non-flat goal subsumptions may generate new variables. These variables are called \emph{decomposition variables}.
				For example if $X \sqsubseteq^? \forall r.Y$ is a goal subsumptions, then the first flattening rule applies and the subsumption is replaced by $X^r \sqsubseteq^? Y$.
				The intended meaning of $X^r$ is that in a solution $\gamma$, $\gamma(X^r) = \{P \mid \forall r.P \in \gamma(X)\}$.
			Decomposition variables are created only when necessary, via the flattening rules; hence, not all variables have decomposition variables.
				For each role name $r$ and a variable $X$, there can be at most one decomposition variable $X^r$,
				sometimes called the $r$-decomposition variable of $X$
			 with 
			 $X$ being  a \emph{parent} of $X^r$.
				
				To enforce the intended meaning for a decomposition variable $X^r$, we create
				an \textbf{increasing subsumption} of the form $X \sqsubseteq^? \forall r. X^r$.
				These subsumptions are kept separately and are not subject to flattening.
				
			However, an increasing subsumption alone does not ensure that, for every particle $\forall r.P $ in the substitution for $X$, 
				$P $ is in the substitution for $X^r$. For this reverse direction, we introduce a  \textbf{decreasing rule}.
				
					We say that a substitution $\gamma$ \emph{obeys the decreasing rule} if and only if for every
					variable $X$ in $\Gamma$, whenever $X^r$ is defined, the following condition holds:
					
					\begin{Rule}\label{rule}
					If a particle of the form $\forall r.P \in \gamma(X)$, then $P \in \gamma(X^r)$.
					\end{Rule}
				
				The decreasing rule is not expressible by a goal subsumption, but is enforced by our algorithm.
				
				Creation of new variables triggers making choices for them as described in the initial steps.
				
%
%
%
%
				
				\item Implicit Solver (Applying rules from Figure~\ref{figure:implicit}). After a flattening step we apply the rules of Implicit Solver are applied exhaustively to all goal subsumptions—whether flat or not yet flat—excluding start and increasing subsumptions. 
				
			\end{enumerate}
			
		\end{enumerate}

		The normalization process terminates in polynomial time with a polynomial increase in the size of the goal. Due to variable choices (such as $\top$, $\bot$, and $A_X$), it is non-deterministic. It produces a normalized goal consisting of the four sets of subsumptions in Figure~\ref{figure:subsumptions}.
		\begin{figure}[H]
			\begin{framed}
		\begin{itemize}
			\item \emph{Solved subsumptions}: they are marked as \emph{solved} by the rules of Implicit Solver and \emph{deleted} from the goal.
			\item \emph{Start subsumptions}:
			$\small\Start = \{ X \sqsubseteq A \mid A_X \text{ is true }\}$ $\cup \{X \sqsubseteq \bot \mid X \text{ is guessed to be a } \bot\text{-}\text{variable} \}$.
			
			\item \emph{Increasing subsumptions}: $\small\{ X \sqsubseteq^? \forall r.X^r \mid \text{ $X^r$ is de-}$ $\text{fined for $X$}\}$.
			\item \emph{Flat subsumptions}:\\
			$\Flat =$ $\{ X_1 \sqcap \cdots \sqcap X_n$ $\sqsubseteq^? Y$ $\mid$ $\text{ all } X_1, \dots, X_n, Y$ $ \textrm{ are variables}\}$ (unsolved subsumptions).
		\end{itemize}
		\end{framed}
		\caption{Normalized goal}\label{figure:subsumptions}
\end{figure}
			The following example illustrates the necessity of the decreasing rule for obtaining correct results.

\begin{example}\label{example:abstr-flatt}
	Let our unification problem contain the goal subsumptions:
	$ \forall rr.\bot\sqsubseteq^? Z,  Z \sqsubseteq^? X, X \sqsubseteq^? \forall r.\bot$.
	The normalized goal is then:\\	
	\emph{start subsumption: } $X^r \sqsubseteq^? \bot$;
	\emph{flat subsumption: }   $Z \sqsubseteq^? X$; the increasing subsumptions are omitted.

	The first start subsumption forces $X^r$ to be a $\bot$-variable, and thus, by the increasing subsumption $X \sqsubseteq^? \forall r.X^r$, $\forall r.\bot$ must appear in the substitution for $X$.
	By the flat subsumption, we know that $\forall r.\bot$ must also appear in the substitution for $Z$.
	However, there is nothing that can force $\bot$ into $Z^r$ unless the decreasing rule is applied.
	If we do apply the decreasing rule, then $\bot$ is forced into $Z^r$—but in that case, we discover that the goal is not unifiable, because $\forall r.\bot \not\sqsubseteq \bot$.
	
	Without the decreasing rule, the following substitution would be incorrectly accepted as a solution:\\
	$
	Z \mapsto \{\forall r.\bot\}, Z^r \mapsto \top, X \mapsto \{\forall r.\bot\}, X^r \mapsto \{\bot\}.
	$
\end{example}

	Next we prove the correctness of the normalization process.
	
	\begin{lemma}\label{lemma:completenessFlattening}(Completeness of normalization)\\
		Let $\Gamma$ be a unification problem, and let $\Gamma'$ be its normalized form obtained through the normalization process described above. If $\gamma$ is a solution of $\Gamma$, then there exists a solution $\gamma'$ of $\Gamma'$ extending $\gamma$ by assigning values to the newly introduced decomposition variables.
	\end{lemma}

	\begin{proof}
		Since we assume $\gamma$ is a ground unifier of $\Gamma$, it can guide the choices in the process of normalization.
		
		\noindent
	The unifier $\gamma$ determines the initial  choices for variables in the following way.
		For a  variable $X$:
		\begin{enumerate}
			\item if $\gamma(X) = \bot$, then $X$ is a $\bot$-variable,
			\item if $\gamma(X) = \top$, then $X$  is a $\top$-variable,
			\item if $A\in \gamma(X)$, then $A_X$ is true otherwise $A_X$ is false.
		\end{enumerate}
		
		Moreover, if a decomposition variable $X^r$ is introduced during the application of a flattening rule, we extend $\gamma$ with the assignment
		$X^r \mapsto \{ P \mid \forall r.P \in \gamma(X) \}$.
		Note that with this extension, $\gamma$ satisfies both the increasing subsumptions and the decreasing rule.
		Also observe that if the set $\{ P \mid \forall r.P \in \gamma(X) \}$ is empty, then $X^r$ should be chosen as a $\top$-variable.
		
		Given that the normalization process terminates, it suffices to prove the lemma for a single step.
		This is shown in the following Claim~\ref{claim}. 
		
			\end{proof}
\begin{claim}\label{claim}		
			If $\Gamma_i$ is a unification problem and $\gamma_i$ is its solution, then either $\Gamma_i$ is already
			normalized or there is a flattening rule applicable to $\Gamma_i$.
		\end{claim}
		\begin{proof}(Proof of the claim)
			Assume that $\Gamma_i$ is not normalized. Hence there is a non-flat, unsolved subsumption in $\Gamma_i$.
			We have two cases to consider.
			\begin{enumerate}
				\item A non-flat subsumption has the form: $s=C_1 \sqcap \cdots \sqcap C_n \sqsubseteq \forall r.P'$. It
				is unified by $\gamma$. We assume that there is no $\bot$ at the top level of $\gamma(C_1 \sqcap \cdots \sqcap C_n)$. Otherwise it would be solved by Implicit Solver,   Figure~\ref{figure:implicit},  rule~\ref{rule-1}.
				
				Notice that every particle $P_1 \in \gamma(\forall r. P')$ is of the form $\forall r. P_2$.
				Hence for each such particle $P_1$, there is a particle $P'_1 \in \gamma(C_1 \sqcap \cdots \sqcap C_n)$, such that $P'_1 \sqsubseteq P_1$.
				Since $P'_1$ cannot be $\bot$, $P'_1 = \forall r.P'_2$ and $P'_2 \sqsubseteq P_2$.
				The first flattening rule tells us to replace $s$ with $s^{-r}$ which will have form:
				$C_1^{-r} \sqcap \cdots \sqcap C_n^{-r} \sqsubseteq^? P'$. 
				
				
				We know that $P'_2 \in \gamma(C_1^{-r} \sqcap \cdots \sqcap C_n^{-r})$, while $P_1 \in \gamma'(P')$ and $P_1$ was chosen to be an arbitrary particle in $\gamma(P')$ ($=\gamma'(P')$). Hence  the extension of $\gamma$ to the decomposition variables,
				satisfies the subsumption $C_1^{-r} \sqcap \cdots \sqcap C_n^{-r} \sqsubseteq P'$, which replaces the original subsumption in the goal.
				
				\item  Now consider a non-flat subsumption:  $s=C_1 \sqcap \cdots \sqcap C_n \sqsubseteq X$. Notice that at this moment we know that $X$ is neither  $\bot$ nor $\top$ in $\gamma$. Let $\gamma(X) = \{P_1, \dots, P_m\}$. The subsumption may not be flat  due to $C_i = \forall r.C'$ in the left hand-side of $s$.
				
				For every role $r \in \fun$, the following is true.
				If there is a particle $P_i \in \gamma(X)$ such that $P_i = \forall r.P'$, then the set of such particles with the top role $r$ is non-empty.
				Then $\gamma$ extended to decomposition variables satisfies the subsumption  $C_1^{-r} \sqcap \cdots \sqcap C_n^{-r} \sqsubseteq X^r$.
				
				If $\gamma(X^r)$ is empty, then $X^r$ should be guessed to be a $\top$-variable and the subsumption is solved by the Implicit Solver,  Figure~\ref{figure:implicit}, rule~\ref{rule-3}.

						Moreover if $A_X$ is true then $A \in \gamma(X)$ and thus $C_1^A \sqcap \cdots \sqcap C_n^A \sqsubseteq A$ is obviously
						unified by $\gamma$, if the original subsumption was.
			\end{enumerate}
			
		\end{proof}
		
		This ends the proof of the claim and thus of Lemma~\ref{lemma:completenessFlattening}.

	Notice that the decreasing rule is not mentioned in the formulation of Lemma~\ref{lemma:completenessFlattening}. This is because if a substitution $\gamma$ 
	is a ground unifier of $\Gamma$, which is not yet extended to the decomposition variables, after extending $\gamma$ by defining
	$\gamma(X^r) := \{P \mid \forall r.P \in \gamma(X)\} $ for each decomposition variable, the decreasing rule is obviously satisfied.
	On the contrary, the decreasing rule must be mentioned in the claim of soundness of the normalization procedure.
	
	\begin{lemma}(Soundness of normalization)\label{lemma:soundness}\\
		If $\Gamma'$ is a normalized unification problem obtained  from $\Gamma$,
		and if $\gamma$ is a unifier of $\Gamma'$ satisfying the decreasing rule for decomposition variables, then it is also a unifier of $\Gamma$.
	\end{lemma}
	
	\begin{proof}
		It is enough to consider the rules of Figure~\ref{figure:flattening}, because the
		reduction rules of Implicit Solver (Figure~\ref{figure:implicit}) are obviously sound.
		Let us assume that in the process of flattening we have obtained a subsumption:
		$s=C_1 \sqcap \dots \sqcap C_n \sqsubseteq P$.
		Let us assume that $\gamma$ unifies this subsumption and it obeys the decreasing rule together with the increasing subsumptions for the decomposition variables. We consider which rule was applied to produce $s$. This rule must have been applied to some subsumption $s'$, which we denote by $s' \to s$.
		We have to show that $s'$ is also unified by $\gamma$. For the inductive argument, we assume that for every $s''$ such that $s' \to s''$, $s''$ is unified by $\gamma$.
		\begin{enumerate}
			\item If the first rule of Figure~\ref{figure:flattening} was applied, then 
			we know that the original subsumption was of the form:
			$s' = C_1' \sqcap \cdots C_k' \sqsubseteq^? \forall r.P$ for a role name $r$ and $k \ge n$.
			
			$s = {s'}^{-r}$. Since $\gamma$ unifies $C_1 \sqcap \dots \sqcap C_n \sqsubseteq P$
			it must also unify $\forall r.C_1 \sqcap \cdots \sqcap \forall r.C_n \sqsubseteq \forall r.P$.
			Hence, $\gamma$ unifies  $\forall r.{C'_1}^{-r} \sqcap \cdots \sqcap \forall r.{C'_k}^{-r} \sqsubseteq \forall r.P$ and
		  because of the monotonicity of subsumption,
			$\gamma$ must unify $s'$ too.
			
			
			\item If $s$ was obtained by the second rule from Figure~\ref{figure:flattening}, from 
			some subsumption  $s' = C'_1 \sqcap \cdots \sqcap C'_n \sqsubseteq^? X$, $k \ge n$,
			then  $s = {s'}^{-r}$ for some role name $r$.
			
			We have to justify that $\gamma$ unifies $s'$.
			By the decreasing rule, we know that 
			$\gamma(X) = \{\forall r.\gamma(X^r) \mid r \in \fun\} \cup \{\alpha \mid \alpha = A ,\text{ if } A_X \text { and } \alpha=\top \text{ otherwise }\}$, where for some $r \in \fun$,
			$\gamma(X^r)$ may be $\top$. Since $s' \to {s'}^{-r}$ for each $r \in \fun$, by the induction assumption,
			$\gamma$ unifies  ${s'}^{-r}$. Since $\gamma$ unifies the increasing subsumptions, it unifies ${C'}_1 \sqcap \cdots \sqcap {C'}_n \sqsubseteq^? \forall r.X^r$ for each $r \in \fun$, while if $A_X$ is true, then ${C'}_1 \sqcap \cdots \sqcap {C'}_n \sqsubseteq^? A$ is also unified by $\gamma$, and
			thus $\gamma(C'_1 \sqcap \cdots \sqcap C'_n) \sqsubseteq^? \gamma(X)$ as required.
		\end{enumerate}
	\end{proof}

		If the set of \Flat is empty, the algorithm terminates with \emph{success}.
	In this case we can construct a unifier $\gamma$ of the goal as shown in Figure~\ref{figure:construction}.
	
		\begin{figure}[H]
			\begin{framed}
	For each variable $X$.
		
		\begin{enumerate}
			\item If the choice for $X$ is $\top$, then define $\gamma(X) := \top$.
			\item If the choice for $X$ is $\bot$, then define $\gamma(X) := \bot$.
			\item If $A_X$ is true, then define $\gamma(X) := \{A\}$.
			\item\label{repeat-step} Repeat until there is no change:
			\begin{itemize}
				\item If $X^r$ is defined  for $X$,
				then redefine $\gamma(X) := \gamma(X) \cup \{\forall r.\gamma(X^r)\}$.
			\end{itemize}
			\item If $\gamma(X)$ is still undefined, set $\gamma(X) := \top$.
		\end{enumerate}
		
		\end{framed}
			\caption{Solution in the case of an empty set of flat subsumptions.}\label{figure:construction}
	\end{figure}
		The initial assignments of $\bot$ and $A$ to variables satisfy the start subsumptions.
The repeat-step (\ref{repeat-step}) will terminate with all increasing subsumptions satisfied.
The solved subsumptions are satisfied by any substitution that respects the choices made when the rules of the Implicit Solver were applied; hence, $\gamma$ is a unifier.

		\section{Shortcuts}\label{section:shortcuts}
		
		If the normalization process have not terminated with \emph{success} or \emph{failure},
		then the set of flat unsolved subsumptions is not empty. We proceed to the next stage, \ie
		computing \emph{shortcuts}.
		
		Informally speaking, a shortcut is a pair of sets of variables from $\Gamma$, $(\Sh,\Smaller)$, where \Sh is non-empty. The set of variables \Sh is called the \emph{main part} of the shortcut and \Smaller--the prefix part.
		
		Here one can see that the concept of shortcut for \flbot is a generalization of a similar concept in the context of \flo, since in the latter case one set of variables would be sufficient, namely the main part, \cite{Borgwardt}.
		
		We will use a shortcut to distribute any particle over the flat subsumptions.
		The particle is placed in the assignment for variables in the main part and some prefixes of this particle should be assigned to
		the variables in the prefix part. Such a distribution of particles, viewed as a substitution that assigns $\top$ to all remaining variables, should unify all flat subsumptions in $\Flat$.
		
		\begin{example} 
			Let the set of flat subsumptions be defined as: $\Flat = \{Y \sqcap X \sqsubseteq^? X^r,\, X^r \sqsubseteq^? X\}$.
			Then an example of a shortcut is $(\Sh, \Smaller)$ where $\Sh = \{X\}$ and $\Smaller = \{X^r, Y\}$.
			If we substitute $X$ with a particle $\forall v.\bot$ (or $\forall v.A$) and assign certain $\bot$-particles in $X^r$ and $Y$ \eg $\forall v'.\bot$,
			where $v'< v$ (or $v' \le v$ in the case of $A$-particle), then obviously, $\Flat$ is unified by such a substitution.
			On the other hand $(\{X^r\}, \{X\})$ is not a shortcut, because substituting $X^r$ with a particle $\forall v.\bot$ and $X$ with $\forall v'.\bot$ with $v' < v$ does not satisfy the subsumption $X^r \sqsubseteq^? X$.
			A good shortcut for this problem is also $(\{X, X^r\}, \emptyset)\}$, because if we substitute $X$ and $X^r$ with the same particle $P$, the flat subsumptions are satisfied.
		\end{example}
		
		Thus the main property of a shortcut $(\Sh, \Smaller)$ is that if the variables in $\Sh$ are substituted with a particle $P$, and the variables in \Smaller with prefixes of $P$ the flat subsumptions $\Flat$ are satisfied.
		We say that a pair $(\Sh, \Smaller)$ satisfies a subsumption $Y_1 \sqcap \cdots \sqcap Y_n \sqsubseteq^? X  \in \Flat$ if the following implications are true:
		\begin{itemize}
		\item	if $X \in \Sh$, then there is $Y_i$ 
		such that $Y_i \in \Sh \cup \Smaller$,
		\item if $X \in \Smaller$, then there is $Y_i$ 
		such that, $Y_i \in \Smaller$.
		\end{itemize}
		%
		Let  $S_{ini}^\bot = \{X \mid X \sqsubseteq^? \bot \in \Start\}$
		and  $S_{ini}^A = \{X \mid A_X \text{ is true }\}$.

		\begin{definition}\label{definition:shortcut}(Shortcut)
			
			A shortcut is a pair $(\Sh, \Smaller)$, where
			\begin{enumerate}
				\item\label{definition:shortcut:condition-disjoint}  $\Sh$, $\Smaller$ are disjoint sets of variables, 
				\item\label{definition:shortcut:condition-smaller} $\Sh\not=\emptyset$,
				\item\label{definition:shortcut:condition-bottom} If  $\Sh \cap  S^\bot_{ini} \not= \emptyset$, then  $\Sh  \subseteq S^\bot_{ini}$ and $\Smaller = \emptyset$,
				\item\label{condition:definition:shortcut} 
				$(\Sh, \Smaller)$ satisfies all flat subsumptions.
			\end{enumerate}
		\end{definition}

		We have two kinds of special shortcuts called \emph{initial}.
		
		\begin{itemize}
			\item (the initial shortcut for bottom) $\inibot = (S_{ini}^\bot, \emptyset)$, 
			\item (an initial shortcut for constant)  $\inicons=(S_{ini}^A, \Smaller)$, where \Smaller is an arbitrary set such that $\Smaller \subseteq S_{ini}^\bot$. 
		\end{itemize}
		
		We do not specify the initial shortcut for the constant explicitly; any shortcut that meets the shortcut conditions will suffice.
	$\bot$-variables may occur in the prefix part of such a shortcut due to Definition~\ref{definition:prefix} of a prefix. The concept $\bot$ is considered a prefix of a constant, whereas $\bot$ itself has no proper prefix.
		
		Condition~\ref{definition:shortcut:condition-disjoint} of Definition~\ref{definition:shortcut} requires that the main part of a shortcut be disjoint from the prefix part. Otherwise, we could have a particle $P$ and its prefix in the same variable $X$. This makes no sense, since $P$ will be removed as redundant.

		
		Condition~\ref{definition:shortcut:condition-bottom} states that the $\bot$-variables occur in main parts of some shortcuts, for which the prefix part must be empty. This is because no
		particles can be proper prefixes of $\bot$.

		A restricted resolving relation is defined between two shortcuts, excluding those whose main components involve bottom variables. We will subsequently extend this relation to cover such cases.
		
		\begin{definition}\label{definition:resolve}(resolving relation)
			Let $s_1 = (\Sh_1, \Smaller_1)$ and $s_2=(\Sh_2, \Smaller_2)$ be two shortcuts.
			We have two cases:
			\begin{enumerate}
			\item  $\Sh_1 \not\subseteq \Sh_{ini}$.
			We say that $s_1$ is restricted-resolved with $s_2$ \wrt a role name $r$ iff the following conditions are
			satisfied.
			\begin{enumerate}
				\item\label{resolve:1} There is a decomposition variable $X^r$ in $\Sh_1$ and for each decomposition variable $Y^r$ in $\Sh_1$, its parent $Y$ is in $\Sh_2$.
				\item\label{resolve:2} For each decomposition variable $Y^r$ in $\Smaller_1$, its parent $Y$ is in $\Smaller_2$.
				\item\label{resolve:3} For each variable $Z$ in $\Sh_2$ such that $Z^r$ is defined, $Z^r$ is in $\Sh_1$.
				\item\label{resolve:4} For each variable $Z$ in $\Smaller_2$ such that $Z^r$ is defined, $Z^r$ is in $\Smaller_1$.
			\end{enumerate}
			\item If $\Sh_1 \subseteq \Sh_{ini}$ we require,  in addition to the above conditions, that 
			 $\Smaller_2 \subseteq \Sh_{ini}$.
		\end{enumerate}
		\end{definition}

%
%
The resolving relation between $s_1$ and $s_2$ with respect to a role name $r$ is denoted by $s_1 \stackrel{r}{\to} s_2$.
Shortcuts with $\bot$-variables in their main part require special treatment for the following reasons:
First, the prefix part of such shortcuts is empty, since $\bot$ has no proper prefixes.
Second, in this case, $\Smaller_2$ must contain only $\bot$-variables, as only such variables can serve as prefixes for particles of role depth 1.

		We say that a shortcut $s = (\Sh, \Smaller)$ \emph{depends on} a shortcut $s' = (\Sh', \Smaller')$, denoted by
		$s \dasharrow s'$,
		if $\Smaller = \Sh' \cup \Smaller'$. In this case we also say that $s'$ \emph{supports} $s$.
		
		For example, if $(\{X\},\{Y,Z\})$ is a shortcut, then it depends on a shortcut $(\{Y\}, \{Z\})$ and this shortcut depends on $(\{Z\}, \emptyset)$.
		Obviously, every sequence of depending shortcuts has to terminate.
		
		\begin{definition}\label{definition:height0}
			A  shortcut $(\Sh, \Smaller)$ is of height $0$ iff $\Sh$ does not contain any decomposition variables.
		\end{definition}
		
	The term \emph{shortcut of height $0$}, as introduced in \cite{Borgwardt}, denotes a shortcut that, when used in constructing a solution, does not produce particles with greater role depth. The height $0$ further indicates that such shortcuts are computed at the initial stage of the overall shortcut computation.
	
		\begin{definition}\label{definition:valid}
			A shortcut $s=(\Sh, \Smaller)$ is \emph{valid} iff
			\begin{enumerate}
				\item 	it is a shortcut of height $0$ or for every decomposition variable $X^r$ in $\Sh$, $s$ is resolved by a valid shortcut \wrt to the role $r$;
				\item $\Smaller= \emptyset$ or $s$ depends on a valid shortcut.
			\end{enumerate}
		\end{definition}
		
	We will show that there is a unifier for an \flbot-unification problem if and only if we can compute
	an acyclic graph of valid shortcuts, starting with the shortcuts of height $0$ and extending it along
	the resolving relation with the set of already computed shortcuts, while containing the initial shortcuts.
		
	To decide unification, our algorithm attempts to detect the existence of such an acyclic structure of shortcuts.
	If no such structure exists, the problem is not unifiable; otherwise, a solution can be computed from it.
		
		\subsection{Level-one shortcuts}
		This section contains some lemmas that are important for the soundness proof in Section~\ref{section:soundness}.
		A \emph{ level-one shortcut} is a shortcut of the form $(\Sh, \Smaller)$ such that $\Sh \not\subseteq S_{ini}^\bot$ and either
		$\Smaller \subseteq S_{ini}^\bot$ or $\Smaller$ is empty.
		
		
		Obviously, for each shortcut $(\Sh,\Smaller)$, if $\Smaller \not\subseteq S_{ini}^\bot$ one can construct a level-one shortcut:
		$(\Sh', \Smaller')$, namely $\Sh' = \Sh \cup \Smaller\setminus S_{ini}^\bot$ and $\Smaller' = \Smaller \cap  S_{ini}^\bot$.
		We call such shortcut the \emph{level-one shortcut for} $(\Sh, \Smaller)$.
		
		To prove the following lemmas, we assume an acyclic structure of valid shortcuts connected by
		the resolving relation. This acyclic structure will be produced by our algorithm, which computes all
		valid shortcuts.
		
		\begin{lemma}\label{lemma:adding-shortcuts}
			Let $(\Sh_1, \Smaller_1)$ and $(\Sh_2, \Smaller_2)$ be two valid shortcuts such that
			$\Smaller_1 = \Sh_2 \cup \Smaller_2$. Then a shortcut $(\Sh_1 \cup \Sh_2, \Smaller_2)$ is also valid.
		\end{lemma}
		
		
		\begin{proof}
			First let us notice that $(\Sh_1 \cup \Sh_2, \Smaller_2)$ satisfies the conditions of Definition~\ref{definition:shortcut}.
			
			
			Since we assume validity, we will use induction on the maximal length $h$ of the path of resolving relations from the shortcut $(\Sh_1, \Smaller_1)$ to a shortcut of height $0$ in an acyclic graph
			of  shortcuts.
			
			If $h=0$, then $\Sh_1$ does not contain any decomposition variables.
			Hence if there are decomposition variables in $\Sh_2$, for each role name $r$,
			by the validity of $(\Sh_2, \Smaller_2)$, there is a valid shortcut $(\Sh_2', \Smaller_2')$ that resolves it \wrt this role.
			
			Obviously, $(\Sh_1\cup\Sh_2,\Smaller_2)$ is resolved by $(\Sh_2', \Smaller_2')$ as well. Hence $(\Sh_1\cup\Sh_2,\Smaller_2)$ is valid in this case.
			
			Now, if $h>0$, then for some role name $r$, a decomposition variable with this role is in $\Sh_1$.
			Because $(\Sh_1, \Smaller_1)$ is valid, there is a valid shortcut $(\Sh_1', \Smaller_1')$ that resolves it \wrt the role $r$.
			
			We have two cases here. 
			If there is no decomposition variable with $r$ in $\Sh_2$, then $(\Sh_1 \cup \Sh_2, \Smaller_2)$ is resolved \wrt the role $r$, with 
			the same shortcut: $(\Sh_1', \Smaller_1')$. Hence we are done for this role name.
			
			If a decomposition variable with $r$ is in $\Sh_2$, by validity of $(\Sh_2, \Smaller_2)$, there is a valid shortcut $(\Sh_2', \Smaller_2')$
			that resolves the one \wrt the role $r$.
			
			Now we compare two shortcuts:  $(\Sh_1', \Smaller_1')$ that resolves  $(\Sh_1, \Smaller_1)$ and $(\Sh_2', \Smaller_2')$ that resolves $(\Sh_2, \Smaller_2)$
			for the same role $r$. Since by assumption $\Smaller_1 = \Sh_2 \cup \Smaller_2$, we see that $\Smaller_1' = \Sh_2' \cup \Smaller_2'$. This follows
			from Definition~\ref{definition:resolve}. 
			Hence by induction $(\Sh_1' \cup \Sh_2', \Smaller_2')$ is a valid shortcut. This shortcut resolves  $(\Sh_1\cup\Sh_2,\Smaller_2)$  \wrt the role $r$ as required.	
		\end{proof}
		
		We use the above lemma in proving the following fact about level-one shortcuts.
		
		\begin{lemma}\label{lemma:level-one-validity}
			Let $(\Sh',\Smaller')$ be the level-one shortcut for $(\Sh, \Smaller)$. 
			If  $(\Sh, \Smaller)$ is valid, then  $(\Sh',\Smaller')$ is also valid.
		\end{lemma}
		
%
		\begin{proof} We will use an induction on  $l$, the length of the path of corresponding shortcuts starting with
			$(\Sh, \Smaller)$. (Let us notice that each such path has to terminate with a shortcut with
			the empty prefix part.)
			
			The base case is when $l=0$. Hence $(\Sh, \Smaller)$ has the form $(\Sh, \emptyset)$. It is its own
			level-one shortcut. Since it is valid by assumption,
			the lemma statement is true.
			
			$l >0$.
			$(\Sh, \Smaller)$ is valid and it depends on a valid shortcut $(\Sh_1, \Smaller_1)$.
			For $(\Sh_1, \Smaller_1)$ the path of corresponding shortcuts is $l-1$.
			
			Hence by induction the level-one shortcut $(\Sh_1', \Smaller_1')$  for $(\Sh_1, \Smaller_1)$ is valid.

			Now we notice that the shortcuts
			$(\Sh, \Smaller)$ and $(\Sh_1, \Smaller_1)$  satisfy the conditions in Lemma~\ref{lemma:adding-shortcuts}, namely
			$\Smaller = \Sh_1' \cup \Smaller'_1$.
			Hence, by Lemma~\ref{lemma:adding-shortcuts} $(\Sh \cup \Sh_1' , \Smaller'_1)$ is a valid shortcut.

			But this  shortcut is the same as the level-one
			shortcut for $(\Sh, \Smaller)$, namely $(\Sh \cup \Sh_1' , \Smaller'_1)= (\Sh',\Smaller')$.
			Hence we have that the level-one shortcut for $(\Sh, \Smaller)$ is valid as required.	
		\end{proof}
Thanks to Lemma~\ref{lemma:adding-shortcuts} we can also prove the following.

		\begin{lemma}\label{lemma:corresponding}
			Let $(\Sh_1, \Smaller_1)$ be resolved with $(\Sh_2, \Smaller_2)$ \wrt role name $r$, $\Smaller_1\not=\emptyset, \Smaller_2\not=\emptyset$ and both these shortcuts are valid.
			Then either $\Smaller_1$ does not contain any $r$-decomposition variable or there is a shortcut $s_3$  supporting  $(\Sh_1, \Smaller_1)$ and $s_4$  supporting
			$(\Sh_2, \Smaller_2)$ such that $s_3$ is resolved with $s_4$ \wrt the role $r$.
		\end{lemma}
		
		\begin{proof}
			Since $(\Sh_1, \Smaller_1)$ and $(\Sh_2, \Smaller_2)$ are valid, there exist
			supporting shortcuts $(\Sh_3, \Smaller_3)$ and $(\Sh_4, \Smaller_4)$.

			\begin{tikzpicture}
				\matrix (m3) [matrix of math nodes, row sep=3em, column sep=4em, minimum width=2em]
				{X \in \Smaller_2 &  (\Sh_2, \Smaller_2) & (\Sh_4, \Smaller_4)\\
					X^r \in \Smaller_1 & (\Sh_1,\Smaller_1) &  (\Sh_3, \Smaller_3)\\};
				\path[-stealth]
				(m3-2-2) edge node [left] {$r$} (m3-1-2)
				(m3-2-3) edge node [left] {$r$} (m3-1-3)
				;
				\path[ dotted]
				(m3-1-1) edge  (m3-1-2)
				(m3-2-1) edge  (m3-2-2);
				\path[ -stealth, dotted]
				(m3-2-2) edge (m3-2-3)
				(m3-1-2) edge (m3-1-3)
				;
			\end{tikzpicture}

			The proof will use induction on the size of prefix parts of the supporting shortcuts $|\Smaller_3| + |\Smaller_4|$.
			
			Assume that there is $X^r \in \Smaller_1$. By Definition~\ref{definition:resolve}, $X \in \Smaller_2$.
			
			If $\Smaller_3$ and $\Smaller_4$ are empty, then 
			all conditions of Definition~\ref{definition:resolve} are satisfied and $(\Sh_3, \Smaller_3)$ is resolved 
			with $(\Sh_4,\Smaller_4)$ \wrt to the role name $r$ as required.
			
			We have the following cases where the conditions required for shortcuts $(\Sh_3, \Smaller_3)$, $(\Sh_4, \Smaller_4)$ may be violated.
			\begin{itemize}
				\item $X^r \in \Sh_3$, but $X \not\in \Sh_4$ ($X \in \Smaller_4$).
				In this case consider a supporting shortcut corresponding to $\Smaller_4$: $(\Sh_5, \Smaller_5)$.
				By Lemma~\ref{lemma:adding-shortcuts}, the shortcut $(\Sh_4 \cup \Sh_5, \Smaller_5)$ is a valid shortcut
				and it is supports $(\Sh_1, \Smaller_1)$. Since size of $\Smaller_5$ is strictly smaller than the size of
				$\Smaller_4$, by induction we have that there are shortcuts $s_3$ and $s_4$ such that $s_3$ is resolved by $s_4$ \wrt the role $r$.
				
				\item $X \in \Sh_4$ and $X^r$ is defined, but $X^r \not\in \Sh_3$ ($X^r \in \Smaller_3$).
				In this case consider a shortcut supporting $\Smaller_3$: $(\Sh_6, \Smaller_6)$.
				By Lemma~\ref{lemma:adding-shortcuts}, the shortcut $(\Sh_3\cup\Sh_6, \Smaller_6)$ is a valid shortcut and 
				it supports $(\Sh_2, \Smaller_2)$. Since the size of $\Smaller_6$ is strictly smaller than the size of 
				$\Smaller_3$, by induction we have that there are shortcuts $s_3$ and $s_4$ such that $s_3$ is resolved by $s_4$ \wrt the role $r$.
			\end{itemize}
			
		\end{proof} 
	\section{Computing shortcuts}\label{section:main}
	Here we present Algorithm~\ref{algorithm:main-new} used to solve a normalized unification problem followed by explanations of its sub-procedures.
		
		\begin{algorithm}[h]
			\caption{Main}\label{algorithm:main-new}
			\begin{algorithmic}[1]
				\Procedure{Main}{$ \Gamma$} \Comment{$\Gamma$ is a normalized unification problem defined \wrt current choices for variables}
				\State $\pred$ is a set of all variables in $\Gamma$
				\State  $\Sh^\bot_{ini} = \{X \in \pred \mid X \sqsubseteq^? \bot \in \Gamma\}$
				\State For each constant $A$, $\Sh^A_{ini} = \{X \in \pred \mid X \sqsubseteq^? A \in \Gamma\}$ 
				\State $\Flat \gets $ flat subsumptions of $\Gamma$
				
				\State\label{main:all-shortcuts-new} $\s \gets$ \Call{AllShortcuts}{$\Gamma$}
				
				\If{$\Sh_{ini}^\bot == \emptyset$}\Comment{\flo-unification}
				
				\If{$\inicons \in \s$, for each constant $A$}\label{algorithm:membership-check1}
				\State \Return success
				\EndIf	
				
				\ElsIf{For each constant $A$, $\Sh_{ini}^A == \emptyset$} \Comment{\pure-unification}
				\If{$\inibot\not\in \s$} 
				\State \Return failure  
				\Else
				\Repeat\label{main:repeat1-new}
				\State $\s \gets$\Call{CheckExistence}{\s}
				\State $\s \gets $\Call{CheckValidity}{$\s$}
				\Until{there is no change in \s}
				\If{$\inibot\in \s$}\label{algorithm:membership-check2} 
				\State \Return success\label{main:success-pure}
				\EndIf
				\EndIf
				\Else \Comment{Full \flbot-unification}
				\If{$\inibot\not\in \s$ or for a constant $A$, $\inicons \not\in \s$}\label{algorithm:membership-check3} 
				\State \Return failure  
				\Else
				\Repeat\label{main:repeat3-new}
				\State $\s \gets$\Call{CheckExistence}{\s}
				\State $\s \gets $\Call{CheckValidity}{$\s$}
				\Until{there is no change in \s}

				\If{ $\inibot \in \s$ and for each constant $A$, $\inicons \in \s$}
				\State \Return success
				\EndIf
				\EndIf
				\EndIf

				\State \Return failure 
				\EndProcedure
			\end{algorithmic}
		\end{algorithm}

		\textbf{Subprocedure \textsc{AllShortcuts}}{
			This procedure runs through all pairs of subsets of variables in \pred, and for each one:
			\begin{enumerate}
				\item Identifies the candidate shortcuts of height $0$ and verifies whether they satisfy the criteria of a shortcut (Definition~\ref{definition:shortcut}), and
					  of height $0$ (Definition~\ref{definition:height0}).
				The shortcuts that pass the checks are added to the set of computed shortcuts.
				
				\item If there are no shortcuts of height $0$ then the algorithms returns \emph{failure}. Otherwise
				it enters a loop, computing next shortcuts based on the already computed ones.
				For each pair $(\Sh, \Smaller)$ of subsets of variables in \pred that is not yet in the set of computed shortcuts
				it checks conditions of Definition~\ref{definition:shortcut},
					and verifies whether there exists an already computed shortcut that resolves $(\Sh, \Smaller)$, as defined in Definition~\ref{definition:resolve}.

				If a shortcut passes the checks, it is added to the set of computed shortcuts. Note that at this stage, shortcuts of heights greater than zero are also added.
				The loop terminates when no new shortcuts can be computed.
			\end{enumerate}
		}
		
		\textbf{Subprocedure \textsc{CheckExistence}}{
			This procedure takes the set of all shortcuts computed and checks if for each one,
			there is a supporting shortcut in the set. Hence it checks if for each shortcut $s_1$, there is
			a shortcut $s_2$, such that $s_1 \dasharrow s_2$. If no $s_2$ is found,  $s_1$ is removed 
			from the set of all computed shortcuts.
			The execution of this procedure ensures that the conditions of Definition~\ref{definition:valid} are satisfied.
		}
		
		\textbf{Subprocedure \textsc{CheckValidity}}{
		After the deletions performed by \textsc{CheckExistence}, this procedure verifies	for each shortcut $s_1$ whether a shortcut of height $0$, say $s_0$, is reachable from $s_1$
		via the resolving relation, i.e.,  $s_1 \stackrel{*}{\to} s_0$.
			If $s_1$ is not resolved in this way, it is removed from the set of computed shortcuts.
		}

		\section{Termination and complexity}\label{section:termination}
		\begin{theorem}
			Let $\Gamma$ be a normalized \flbot unification problem. 
			Algorithm~\ref{algorithm:main-new} terminates in at most exponential time in the size of $\Gamma$.
		\end{theorem}
		
%
%
%
%
		\begin{proof}
		We assume that our problem is normalized and the sets of flat and start subsumptions are non-empty.
			
			The algorithm computes all shortcuts (line \ref{main:all-shortcuts-new}). This step is  exponential, because there are only exponentially many possible shortcuts. 
			
			Now the algorithm detects one of the 3 cases:
			\flo-unification case, 	\pure-unification case, or full \flbot case.
			
			The first case is the simplest one, but otherwise they contain similar steps.
			
			If a required initial shortcut is not computed, the algorithm terminates with failure. 
			
			Otherwise, in the first case it returns success, in the two remaining cases, it proceeds to the next step.
			
			The next repeat-loop (line~\ref{main:repeat1-new},  \ref{main:repeat3-new}) causes deletions in the set of already computed shortcuts and hence it can be executed only at most exponentially many times. Each time it has to perform two sequential checks, each of them costing exponentially many  steps.
			
			After the loop terminates, the algorithm terminates.
			Hence overall the time needed for the algorithm to terminate is exponential in the size of the problem.
		\end{proof}

		\section{Soundness}\label{section:soundness}
		
		
		
		\begin{theorem}\label{theorem:soundness}
			Let $\Gamma$ be the normalized unification problem and 
			let the main algorithm  on $\Gamma$  terminate with success, then there exists
			a ground \pure-unifier $\gamma$ of $\Gamma$.
		\end{theorem}
		
		\begin{proof}
			
			We assume that the algorithm terminated with success. 
%
			Since the problem is normalized  and the algorithm has run on  $\Gamma$, we can assume that the set of start subsumptions  and 
			the set \Flat of flat subsumptions is not empty.
			
			The algorithm has computed shortcuts \s in a special order, starting from the shortcuts of height $0$. We can see this as if the algorithm has produced a directed acyclic graph
			defined on the set of all computed shortcuts. The graph is $(V, E)$, where $V = \s$, and $E$ -- the set of edges given by the resolving relation (Definition~\ref{definition:resolve}) between shortcuts.
			
			The algorithm provides additional connections between the corresponding shortcuts in the following way:
			if $s=(\Sh, \Smaller)$ and $s' = (\Sh', \Smaller')$, where $\Smaller = \Sh' \cup \Smaller'$, then we have an additional edge $s \dotarrow{} s'$. A dotted edge indicates where the  $\bot$-particles in a prefix-part of a shortcut, should be created.
			
			Now let us take a  connected sub-graph of this construction containing the initial shortcuts.
			We show that a unifier can be constructed based on this graph.
			The construction begins with the substitution $\gamma_0 = \{X \mapsto \top \mid X \in \pred\}$  and proceeds by extending it with particles, until a unifier is obtained.
			For the description we use the following notions.
			
				We say that a ground particle $P$ of the role depth $i$ \emph{is created in a shortcut} $S = (\Sh,\Smaller)$ at step $i$,
				if an extended substitution $\gamma_i$ is defined such that,  for every variable $X \in \Sh$, 
				$\gamma_i(X) = \gamma_{i-1}(X) \cup \{P\}$.	
				%
				%
				If $X^r \in \Sh$, then we say that $P$ is \emph{active} in $S$ at step $i$.
				\emph{Step $i$} of the construction consists of the creation of all the particles of the role depth $i$.
			%
			Hence, the construction of a unifier consists of \emph{deactivating} active particles of role depth $i-1$.

			\textbf{Invariant}
				The invariant of the construction is that at any time the current substitution satisfies
				the flat subsumptions of $\Gamma$.

			The invariant is obviously true for the initial substitution $\gamma_0$.
		%
			The following condition has to be satisfied for the invariant to be true at any step.

			\textbf{Condition}\label{soundness:condition}
				Let  $S = (\Sh,\Smaller)$ be a shortcut and $S' = (\Sh', \Smaller')$ be a shortcut corresponding to $\Smaller$,
				\ie $\Smaller = \Sh' \cup \Smaller'$.
				
				We allow a particle $P$ to be created in $S$ at step~$i$ only if a prefix of $P$ was already created in $S'$ at a step smaller than~$i$.

			The above Condition forces us to construct a solution in a \emph{bottom-up} way.
			
			\textbf{Steps of the construction}
			\begin{enumerate}
				\item (Step 0)
				In order to satisfy the start subsumptions, we create the bottom particle $\bot$ in the initial shortcut $\inibot = (\Sh^\bot_{ini}, \emptyset)$ and likewise,
				we  initialize the $A$-variables in $\inicons = (\Sh^A_{ini}, \emptyset)$.
				Observe that such an extension of $\gamma_0$ to $\gamma_1 = \{[X \mapsto \gamma_0(X)\cup \{\bot\}] \mid X \in S^\bot_{ini}\}
				\cup \{[X \mapsto \gamma_0(X)\cup \{A\}] \mid X \in S^A_{ini}\}$ satisfies the invariant.
				Some particles become active at this step in \inibot or \inicons. If there are no active particles, we terminate
				with success. The current substitution is a unifier.

				\item 	(Step 1) In order to \emph{deactivate} the active particles, we follow the labeled arrows in the graph of shortcuts, representing how the shortcuts are resolved with each other.\\
				

				At this step one can choose to create arbitrarily particles of the form $\forall r.\bot$ in level-one shortcuts of the form $(\Sh, \emptyset)$, but only if no variable in $\Sh$ has $r$-decomposition variable defined.
				In other words, we can extend $\gamma(Z)$ with $\forall r.\bot$, where $Z \in \Sh$, only if $Z^r$ is not defined. If $Z^r$ is defined, the shortcut   $(\Sh, \emptyset)$ cannot be used for creation of $\forall r.\bot$.
				Otherwise, the decreasing rule could be violated by such an extension.

				\item 	(Step i+1) Now assume that the substitution $\gamma$ is defined for the variables, assigning to them sets of particles of the role depth at most $i$.  Let us assume that $s = (\Sh, \Smaller)$ and a particle $P$ of
				depth $i$ is active in $s$. Hence $P$ was created in $s$ at step $i$ and it is added to the substitution
				for a decomposition variable $X^r$ in $\Sh$. ($P \in \gamma_i(X^r)$)
				Since the particles of depth $i+1$ are not yet created, the increasing subsumption: $X \sqsubseteq^? \forall r.X^r$ is not satisfied.
				
				The algorithm gives us an arrow from $s$ to $s'=(\Sh', \Smaller')$ with the label $r$, $s \stackrel{r}{\to} s'$.
				Hence $X \in \Sh' $. We create a particle $\forall r.P$ in $s'$, extending the substitution to $\gamma_{i+1}(Y) = \gamma_i(Y) \cup \{\forall r.P\}$ for every $Y \in \Sh'$.
				In this way the particle $P$ is deactivated in $s$ and the increasing subsumption  $X \sqsubseteq^? \forall r.X^r$ is satisfied \wrt to this particle ($\forall r.P \in \gamma_{i+1}(X)$ and $P \in \gamma_{i+1}(X^r)$).		
			\end{enumerate}\end{proof}	
		The following example illustrates a graph of shortcuts. One can construct a unifier using the construction
		in the proof of Theorem~\ref{theorem:soundness}.
		\begin{example}
		Let the normalized goal be the following.
		
		Flat subsumptions:
		 $\{ Y^r \sqsubseteq^? Z, U\sqcap Y \sqsubseteq^? Z^s\}$,\\
		Start subsumptions: $\{Y^r \sqsubseteq^? A, U^{rs} \sqsubseteq^? \bot\}$.
		
		The algorithm yields the following shortcuts:
		\inibot = $(\{U^{rs}\}, \emptyset)$, \inicons = $(\{Y^r\}, \{U^{rs}\})$,
		 $s_1=(\{U^r, Y^r\}, \{U^{rs}\})$, $s_2 = (\{U, Y, Z^s\}, \emptyset)$, $s_3 = (\{Z\}, \{Y^r, U^{rs}\})$. Notice that $s_3$ is of height $0$ and it depends on \inicons.
				\inicons is resolved by $s_2$.
				The graph of resolving relations between the shortcuts is presented below.
%
	
	\begin{tikzpicture}
		\matrix (m0) [matrix of math nodes, row sep=1em, column sep=4em, minimum width=2em]
		{ 
			\inibot & s_1 & s_2 & s_3 \\
			\inicons &     &     &     \\
		};
		
		\path[-stealth]
		(m0-1-1) edge node [above] {$s$} (m0-1-2)
		(m0-1-2) edge node [above] {$r$} (m0-1-3)
		(m0-1-3) edge node [above] {$s$} (m0-1-4)
		(m0-2-1) edge node [below] {$r$} (m0-1-3);
		
		\path[dotted, -stealth]
		(m0-1-2) edge[bend left] (m0-1-1)
		(m0-1-4) edge[bend left] (m0-2-1);
	\end{tikzpicture}%
%
		\end{example}
		Next we prove that Condition of the construction can be satisfied at each step and hence the construction in the proof of Theorem~\ref{theorem:soundness} is always possible.
		
		\begin{lemma}
			Let $P$ be a particle that should be created in $s = (\Sh, \Smaller)$ at some step $i$.
			
			Then the prefixes of $P$ have been created at previous steps at the variables in $\Smaller$, so that 
			Condition of the construction  is satisfied.
		\end{lemma}
		\begin{proof}(sketch)
		The proof proceeds by induction on 
	$i$, where 	$i$ is the index of a step in the construction. This index also corresponds to the role depth of every particle created at step 	$i$.
		For the proof, we use Lemma~\ref{lemma:level-one-validity} and Lemma~\ref{lemma:corresponding} to demonstrate how the prefixes are formed.

		\end{proof}

		The next lemma shows that the construction in the proof of Theorem~\ref{theorem:soundness} terminates.
		
		\begin{lemma}
			Let $\Gamma$ be a normalized unification problem. If the algorithm terminates with success,
			then the construction of a unifier terminates.
		\end{lemma}
		
		\begin{proof}
		The reason is that the graph of shortcuts is acyclic, so every path in this graph has finite length—no longer than the total number of shortcuts.
		Therefore, the construction of a solution terminates after creating only finitely many particles, each with a role depth at most exponential in the size of $\Gamma$.
		\end{proof}

		\section{Completeness}\label{section:completeness}
		
		\begin{theorem}\label{theorem:completeness}
			Let $\Gamma$ be a  normalized unification problem and let $\gamma$ be a ground unifier of $\Gamma$
			obeying the decreasing rule.
			Then the main algorithm will terminate with success. 
		\end{theorem}
		
		
		\begin{proof}
			
			The proof demonstrates how the solution $\gamma$ guarantees a non-failing run of the algorithm. We assume that $\gamma$ is reduced \wrt to the properties of \flbot and is minimal in size; that is, the sum of the role depths of the particles in the range of $\gamma$ is minimal.

			We identify  the set of $\bot$-variables as 
			$\Sh_{ini}^\bot = \{ X \in \pred \mid \gamma(X) = \bot\}$ 
			and the set of variables containing constant $A$ as
			$\Sh_{ini}^A = \{X \in \pred \mid A \in \gamma(X)\}$
			
			If all these sets are empty, we can construct a solution sending all variables to $\top$.
			
			Otherwise, we identify shortcuts defined by $\gamma$ in the following way:
			for a given particle $P$ in the range of $\gamma$, there is a shortcut, $(\Sh_P, \Smaller_P)$ where:
			
			$\Sh_P = \{X \in \pred \mid P \in \gamma(X)\}$,
			
			$\Smaller_P = \{Y \in \pred \mid P' \in \gamma(Y),  P' \text{ is a prefix of } P\}$.
			
			For a shortcut  $(\Sh_P, \Smaller_P)$, $\Smaller_P$ satisfies the conditions of  Definition~\ref{definition:shortcut}.
			The algorithm calls the sub-procedure \textsc{AllShortcuts} which computes all possible shortcuts defined in this way.
			First let us notice that there will be shortcuts of height $0$, namely the shortcuts defined
			for the particles of maximal role depth in $\gamma$. Hence, the set of computed shortcuts will not be empty.
			
			Now it is easy to see that since $\gamma$ satisfies increasing subsumptions and obeys the decreasing rule, then
			in the set of all shortcuts defined by $\gamma$, there will be shortcuts properly resolved between themselves and thus they will be computed by  \textsc{AllShortcuts} and will not be deleted during the run of the algorithm.
			
%
%
			Hence, for each particle $P$ in the range of $\gamma$,  a shortcut defined by $\gamma$ will be computed and not deleted. Now we consider the
			particular cases determined by the initial shortcuts.
			Here we justify completeness only in the most interesting case, \ie when
			$\Sh_{ini}^\bot$ is empty. In this case we do not have any $\bot$-variables. 
				If there are no $\bot$-particles in the range of $\gamma$, then $\gamma$ is an \flo-unifier.
				There are no prefixes for the $A$-particles, hence in each shortcut, the prefix part is empty. The initial shortcut $(\Sh_{ini}^A, \emptyset)$ is computed by  \textsc{AllShortcuts}. This clearly shows that the algorithm is a generalization of the unification algorithm for \flo.
				
				If there are $\bot$-particles in the range of $\gamma$, we show that there is a smaller unifier $\gamma'$, which contradicts
				our assumption about minimality of $\gamma$. 

				For that, let us consider a $\bot$-particle $P_{min}$ of minimal role depth assigned to a variable by $\gamma$. We know that it is not $\bot$.
				Now, consider the shortcut $(\Sh_{min}, \Smaller_{min})$ defined by $\gamma$ for $P_{min}$. Let us notice that
				in this case $\Smaller_{min}$ must be empty, since otherwise $\gamma$ assigns some prefixes of $P_{min}$ to some variables, which
				contradicts minimality of $P_{min}$.
				
				Now assume w.l.o.g. that the particle $P_{min}$ is of the form $\forall r.P'$. It implies that all variables in $\Sh_{min}$ do not have any $r$-decomposition variables defined. Otherwise, since $\gamma$ obeys the decreasing rule, we would have a shortcut for $P'$ and $P'$ is smaller than $P_{min}$.
				
				Hence, $(\Sh_{min}, \emptyset)$ is a resolved shortcut that has no $r$-incoming edges in the graph of shortcuts. Notice also that the variables in $\Sh_{min}$ are not $\bot$-variables, since
				they contain $P_{min}$ under the substitution $\gamma$.
				
				We obtain a different substitution $\gamma'$ by creating $\bot$ in $(\Sh_{min}, \emptyset)$. 
				This changes all variables in $\Sh_{min}$ to $\bot$-variables. 
				Now, we create new substitution $\gamma'$ using the old structure of shortcuts defined by $\gamma$.
				Notice that now Condition~\ref{definition:shortcut:condition-bottom} of Definition~\ref{definition:shortcut} is not satisfied, because  $\bot$-variables from $\Sh_{min}$ may occur in the main part of some shortcuts.

				In order to satisfy the increasing subsumptions, we use the same paths of resolving relation as in the case of $\gamma$ starting with the
				shortcut $(\Sh_{min}, \emptyset)$ as the new \inibot.
				
				The construction from the proof of Theorem~\ref{theorem:soundness} goes through, although the new substitution $\gamma'$ may be not reduced.
				After reductions, $\gamma'$ is strictly smaller than $\gamma$, since at least the $\bot$-particles assigned by $\gamma$ to variables in $\Sh_{min}$, in $\gamma'$ are replaced by $\bot$. Other particles
				might also be reduced by introducing $\bot$ in $\gamma'$.
				
				The construction is illustrated in Example~\ref{example}.
				
					The conclusion of this argument is that if a solution to a problem uses $\bot$ in its particles,
				the set of $\bot$-variables must be non-empty or else the solution is not minimal.
				
				\end{proof}
				%
				%
				
				
				\begin{example}\label{example}
					Let the goal be $\Gamma = \{X \sqsubseteq^? \forall r.Y, Y \sqsubseteq^? \forall r.X, X \sqsubseteq^? \forall r.A\}$
					Let the normalized $\Gamma$ be $\{X^r \sqsubseteq^? Y, Y^r \sqsubseteq^? X, X^r \sqsubseteq^? A\}$ with increasing subsumptions omitted.
					Let $X^r \sqsubseteq^? A$ be a start subsumption.
					
					Now let a solution $\gamma$ be:
					$[X^r \mapsto A \sqcap \forall r.\bot, Y^r \mapsto \forall r.\bot, X \mapsto \forall r.A \sqcap \forall rr.\bot, Y \mapsto \forall rr.\bot]$.
					
					One can see that $\gamma$ is a solution. It satisfies the increasing subsumptions and the decreasing rule.
					
%
					%
					For $\forall r.\bot$, $\gamma$ defines shortcut $(\{X^r, Y^r\}, \emptyset)$ and this is $(\Sh_{min}, \emptyset)$ from the above argument.
					%
					%
					%
					If $\bot$ is created in $(\{X^r, Y^r\}, \emptyset)$, the
					structure can be reused to generate $\gamma'$.
					
					$[X^r \mapsto A \sqcap \bot, Y^r \mapsto \bot, X \mapsto \forall r.A \sqcap \forall r.\bot,
					Y \mapsto \forall r.\bot]$.
					
					After removing redundant particles, we have a smaller solution with $\bot$-variables: $X^r$ and $Y^r$. 
%

				\end{example}

		\section{Conclusions}\label{section:conclusions}
		

		The procedure presented in this paper shows that the \flbot-unification is
		in the ExpTime-class of problems. Theorem 1 in \cite{BaFe-UNIF-22} further shows that our problem is also ExpTime-hard.
		
		The algorithm described here is a generalization  of the unification algorithm that can be 
		extracted from \cite{Borgwardt} for the description logic \flo. A more direct exposition of that
		algorithm may be found in \cite{Morawska2025a}.
		The result  extends the area of description logics for which unification is known,
		\ie the description logic \el and \flo. 
		
		One can attempt to extend similar techniques to the logic $\mathcal{FLE}$, which combines \flo with \el.  However, the main challenge remains the logic \alc, which adds negation to the constructors and is therefore a full Boolean description logic. Very little is known about unification in \alc, except that it is undecidable in an extension of this logic \cite{Wolter2008}.
		
		Our aim is to implement the algorithm presented in this paper and perform experiments that may lead to useful applications in the area of knowledge representation.
	\bibliographystyle{kr}
	\bibliography{FLbottom}
\end{document}